\documentclass[reqno,tbtags,intlimits,a4paper,oneside,12pt]{amsart}
\usepackage[cp1251]{inputenc}%
\usepackage[english]{babel}
\usepackage{amssymb,upref,calc,url}
\usepackage[mathscr]{eucal}
\usepackage{graphicx}
\usepackage{amsmath,amscd,amsthm,verbatim}
\usepackage[all]{xy}
\usepackage[in]{fullpage}
\newtheorem{thm}{\hskip\parindent Theorem}[section]
\newtheorem{lem}[thm]{\hskip\parindent Lemma}
\newtheorem{prb}[thm]{\hskip\parindent Problem}
\newtheorem{cor}[thm]{\hskip\parindent Corollary}

\theoremstyle{definition}
\newtheorem{dfn}{\hskip\parindent Definition}[section]
\newtheorem{ex}[dfn]{\hskip\parindent Example}
\newtheorem{rem}[dfn]{\hskip\parindent Remark}
\DeclareMathOperator{\wt}{wt}

\newcommand{\CU}{\mathcal{C}(z,\lambda)}

\newcommand{\Sym}{\mathop{\rm Sym}}
\newcommand{\Sh}{\mathop{\rm Sh}}
\newcommand{\ShW}{\mathop{\rm ShW}}
\newcommand{\bx}{\mathbf{x}}
\newcommand{\be}{\mathbf{e}}
\newcommand{\bp}{\mathbf{p}}

\begin{document}

\title{Hyperelliptic Sigma functions\\ and \\Adler--Moser polynomials}
\author{V.\,M.~Buchstaber, E.\,Yu.~Bunkova}
\address{Steklov Mathematical Institute of Russian Academy of Sciences, Moscow, Russia}
\email{buchstab@mi-ras.ru, bunkova@mi-ras.ru}

\begin{abstract}
In a 2004 paper by V.\,M.\,Buchstaber and  D.\,V.\,Leykin, published in ``Functional Analysis and Its Applications,'' for each $g > 0$, a system of $2g$ multidimensional heat equations in a~nonholonomic frame was constructed. The sigma function of the universal hyperelliptic curve of genus $g$ is a solution of this system.
In our previous work, published in ``Functional Analysis and Its Applications,'' explicit expressions for the Schr\"odinger operators that define the equations of the system considered were obtained in the hyperelliptic case.

In this work we use these results to show that if the initial condition of the system considered is polynomial, then the solution of the system is uniquely determined up to a constant factor.
This has important applications in the well-known problem of series expansion for the hyperelliptic sigma function.
We give an explicit description of the connection of such solutions to well-known Burchnall--Chaundy polynomials and Adler--Moser polynomials. We find a system of linear second-order differential equations that determines the corresponding Adler--Moser polynomial.
\end{abstract}

\maketitle

\section{Introduction} \label{S0}
 
Let $g \in \mathbb{N}$. For a meromorphic function $f$ in $\mathbb{C}^g$, a vector $\omega \in \mathbb{C}^g$ is a \emph{period} if $f(z+\omega) = f(z)$ for all~$z \in \mathbb{C}^g$.
If a meromorphic function $f$
has $2g$ independent periods in~$\mathbb{C}^g$, then $f$ is called an \emph{Abelian function}.
Thus, an Abelian function is a meromorphic~function on the complex torus $T^g = \mathbb{C}^g/\Gamma$,
where $\Gamma$ is the lattice formed by the periods.

We work with the universal hyperelliptic curve of genus $g$ in the model
\[
\mathcal{V}_\lambda = \{(x, y)\in\mathbb{C}^2 \colon
y^2 = x^{2g+1} + \lambda_4 x^{2 g - 1}  + \lambda_6 x^{2 g - 2} + \ldots + \lambda_{4 g} x + \lambda_{4 g + 2}\}. 
\]
Each curve is defined by specialization of parameters $\lambda = (\lambda_4, \lambda_6, \ldots, \lambda_{4 g}, \lambda_{4 g + 2}) \in \mathbb{C}^{2 g}$.
Let $\mathcal{B} \subset \mathbb{C}^{2g}$ be the subspace of parameters such that the curve $\mathcal{V}_{\lambda}$ is nonsingular for~$\lambda \in~\mathcal{B}$.
Then we have $\mathcal{B} = \mathbb{C}^{2g} \backslash \Sigma$, where $\Sigma$ is~the discriminant hypersurface of the~universal curve.

For each $\lambda \in \mathcal{B}$, the set of periods of holomorphic differentials on the curve $\mathcal{V}_\lambda$ generates a lattice $\Gamma_\lambda$ of rank $2 g$ in $\mathbb{C}^g$. 
A \emph{hyperelliptic function of genus} $g$ (see \cite{B2}, \cite{BEL-12}, and \cite{BEL18}) is a~meromorphic function on $\mathbb{C}^g \times \mathcal{B}$
such that, for each $\lambda \in \mathcal{B}$, it's restriction on $\mathbb{C}^g \times \lambda$
is an Abelian function. Here the torus~$T^g$ is~the~Jacobian variety $\mathcal{J}_\lambda = \mathbb{C}^g/\Gamma_\lambda$ of the curve~$\mathcal{V}_\lambda$.
We denote by $\mathcal{F}$ the field of hyperelliptic functions of genus $g$. For the properties of this field, see \cite{BEL-12} and \cite{BEL18}. 

We use the theory of hyperelliptic Kleinian functions (see \cite{BEL-12}, \cite{Baker}, \cite{BEL}, \cite{BEL-97}, and~\cite{WW} for elliptic functions).
Take the coordinates
$(z, \lambda)$
in $\mathbb{C}^g \times \mathcal{B} \subset \mathbb{C}^{3g}$.
Let $\sigma(z, \lambda)$ be the hyperelliptic sigma function (or the elliptic sigma function in the case of the genus $g=1$). We denote~$\partial_k = {\partial \over \partial z_k}$.
Following \cite{B2}, \cite{BEL18}, and \cite{B3}, we use the notation
\[
\zeta_{k} = \partial_k \ln \sigma(z, \lambda), \qquad
\wp_{k_1, \ldots, k_n} = - \partial_{k_1} \cdots \partial_{k_n} \ln \sigma(z, \lambda),
\]
where $n \geqslant 2$ and $k_s \in \{ 1, 3, \ldots, 2 g - 1\}$.
The functions $\wp_{k_1, \ldots, k_n}$ give examples of~hyperelliptic functions. 
The field $\mathcal{F}$ is the field of fractions of the polynomial ring~$\mathcal{P}$ generated by the functions $\wp_{k_1, \ldots, k_n}$, where $n \geqslant 2$ and $k_s \in \{ 1, 3, \ldots, 2 g - 1\}$. See also \cite{BMulti}.

Note that we denote the coordinates in~$\mathbb{C}^g$ by~$z = (z_1, z_3, \ldots, z_{2g-1})$. The indices of coordinates $z = (z_1, z_3, \ldots, z_{2g-1}) \in \mathbb{C}^g$ and of parameters $\lambda = (\lambda_4, \lambda_6, \ldots, \lambda_{4 g}, \lambda_{4 g + 2}) \in \mathbb{C}^{2 g}$ determine their weights as
$\wt z_k = - k$ and $\wt \lambda_k = k$.
For suitable weights of all the~other variables all the equations in this paper are of homogeneous weight.

We note a property of the hyperelliptic sigma function
$\sigma(z, \lambda)$, that holds as well for the~more general case of $(n,s)$-curves \cite{ns}: at $z=0$ it can be decomposed as a homogeneous series of degree $-\frac{1}{2}g(g+1)$ in $z$ with polynomial coefficients in $\lambda$. We have
\[
\sigma(z, \lambda) = \sum_{|I|>0}\sigma_I(\lambda)z^I,
\]
where $I = (i_1,i_3, \ldots, i_{2g-1}) \in \mathbb{Z}_\geqslant^{g}$, $z^I = z_1^{i_1}\cdots z_{2g-1}^{i_{2g-1}}$, and $\sigma_{I}(\lambda) \in \mathbb{Q}[\lambda]$.
In this work we~normalize the hyperelliptic sigma function by the condition
\[
\sigma(z, \lambda) = z_1^{\frac{1}{2}g(g+1)} + \ldots 
\]

In the case $g=1$ we have the elliptic Weierstrass sigma function with $g_2 = -4\lambda_4$ and~$g_3 = -4\lambda_6$:
\[
\sigma(z, \lambda) = z_1 - \frac{g_2}{2}\,\frac{z_1^5}{5!} - 6g_3\frac{z_1^7}{7!} + \ldots
\]

In the case $g=2$ we have the hyperelliptic Kleinian sigma function:
\[
\sigma(z, \lambda) = z_1^3 - 3 z_3 + {1 \over 420} \lambda_4 z_1^7 + {1 \over 4} \lambda_4 z_1^4 z_3 - {1 \over 1890} \lambda_6 z_1^9 + {1 \over 30} \lambda_6 z_1^6 z_3 + {1 \over 2} \lambda_6 z_1^3 z_3^2 - { 1 \over 2} \lambda_6 z_3^3 + \ldots
\]

For $g=1$ and $2$ there exist efficient algorithms for recursive calculation of coefficients $\sigma_I(\lambda)$ of the series $\sigma(z, \lambda)$,
see \cite{BL-05}, \cite{4A}. For any $g\geqslant 1$ the construction of the series~$\sigma(z, \lambda)$ is known, see \cite{Nak-10}. This construction may be implemented for any curve  $\mathcal{V}_\lambda$, including singular curves, see \cite{Nonhol}, \cite{BL-05}.

In the case of the most singilar curve, namely the curve
\[
\mathcal{V}_0 = \{(x, y)\in\mathbb{C}^2 \colon y^2 = x^{2g+1} \}, 
\]
the sigma function
is given by the homogeneous polynomial $\sigma(z, 0)$, that is called \emph{the rational limit of the sigma function}.

Let us define the \emph{polynomial Lie algebra of vector fields tangent to the discriminant hypersurface} $\Sigma$ \emph{in}~$\mathbb{C}^{2g}$.
We denote it by~$\mathscr{L}_{L}$.
For this polynomial Lie algebra \cite{BPol},
the generators $\{L_{0}, L_{2}, L_{4}, \ldots, L_{4 g - 2}\}$ are the vector fields 
\[
 L_{2k} = \sum_{s=2}^{2g+1} v_{2k+2, 2s-2}(\lambda) {\partial \over \partial \lambda_{2k}},
\]
where $v_{2k+2, 2s-2}(\lambda) \in P$,
and $P$ is the ring of polynomials in $\lambda \in \mathcal{B} \subset \mathbb{C}^{2g}$.
At a point~$\lambda \in \mathcal{B}$ these vector fields determine a $2g$-dimensional nonholonomic frame.

The structure of a Lie algebra as a $P$-module with
generators
$
1, \; L_{0}, \; L_{2}, \; L_{4}, \; \ldots, \; L_{4 g - 2}
$ is determined by the polynomial matrices $V(\lambda)=(v_{2i,2j}(\lambda)),$ where
$i,j=1,\dots,2g$, and~$C(\lambda) = (c_{2i,2j}^{2k}(\lambda))$, where $i,j,k=0,\dots,2g-1$, such that 
\[
[L_{2i},L_{2j}]=\sum_{k=0}^{2g-1} c_{2i,2j}^{2k}(\lambda)L_{2k},
\quad[L_{2i},\lambda_{2q}]=v_{2i+2,2q-2}(\lambda),
\quad[\lambda_{2q},\lambda_{2r}]=0.
\]
Here $\lambda_q$ is the operator of multiplication 
by
the function $\lambda_q$ in $P$.

For the Lie algebra $\mathscr{L}_{L}$,
explicit expressions for the matrix $V(\lambda)$ 
can be found in Section 4.1 of \cite{A} (see also~\cite{BPol} and Lemma~3.1 in \cite{4A}). 
The elements of this matrix are given by the following formulas. For~convenience, we~assume that $\lambda_s = 0$ for all $s \notin \{0,4,6, \ldots, 4 g, 4 g + 2\}$ and~$\lambda_0 = 1$.
For $k, m \in \{ 1, 2, \ldots, 2 g\}$, if $k\leqslant m$, then we set
\[
 v_{2k, 2m}(\lambda) = \sum_{s=0}^{k-1} 2 (k + m - 2 s) \lambda_{2s} \lambda_{2 (k+m-s)}
 - {2 k (2 g - m + 1) \over 2 g + 1} \lambda_{2k} \lambda_{2m},
\] 
and if $k > m $, then we set $v_{2k, 2m}(\lambda) = v_{2m, 2k}(\lambda)$.
The structure polynomials $c_{2i,2j}^{2s}(\lambda)$ are described in Theorem~2.5 of~\cite{Nonhol}.

The vector field $L_0$ is the Euler vector field; namely, since $\wt \lambda_{2k} = 2 k$, we have
\begin{align*} 
 [L_0, \lambda_{2k}] &= 2 k \lambda_{2k}, & [L_0, L_{2k}] &= 2 k L_{2k}.
\end{align*}
This determines the weights of the vector fields $L_k$, namely, $\wt L_{2k} = 2 k$.

The classical \emph{Lie--Witt algebra} $W_0$ (see \cite{BM}) over the field $\mathbb{C}$ of complex numbers is generated by the operators $l_{2i}$, where $i = 0, 1, 2, \ldots$, with the commutation relations
\[
 [l_{2i}, l_{2j}] = 2 (j-i) l_{2 (i+j)}.
\]
With respect to the bracket $[\cdot,\cdot]$ the~Lie--Witt algebra $W_0$ is generated by the three operators $l_0$, $l_2$, and $l_4$.
The graded polynomial Lie algebra $\mathscr{L}_{L}$
over $P$ is a deformation of the~Lie--Witt algebra $W_0$. It is also generated by only three operators,  $L_0$, $L_2$, and~$L_4$. The following relation holds (see Lemma 3.3 in \cite{BB20}):
\[
[L_2, L_{2k}] = 2 (k-1) L_{2k+2} + {4 (2 g - k) \over (2 g + 1)} \left( \lambda_{2k+2} L_0 - \lambda_4 L_{2k-2}\right). 
\]

Now we introduce the \emph{Schr\"odinger operators}. We consider the space $\mathbb{C}^{3g}$ with coordinates $(z, \lambda)$
and let $\CU$ denote the ring of differentiable functions in $z$ and
$\lambda$. We~set
\begin{equation} \label{HQ}
Q_{2k} = L_{2k} - H_{2k}, \qquad k = 0, 1, 2, \ldots, 2g-1,
\end{equation}
where
\begin{equation} \label{e2}
H_{2k} = {1 \over 2} \sum \left( \alpha_{a,b}^{(k)}(\lambda)\partial_a \partial_b + 2 \beta_{a,b}^{(k)}(\lambda)z_a\partial_b + \gamma_{a,b}^{(k)}(\lambda)z_a z_b\right) + \delta^{(k)}(\lambda),
\end{equation}
the summation is over odd $a$ and $b$ from $1$ to $2g-1$, and $\alpha_{a,b}^{(k)}(\lambda)$, $\beta_{a,b}^{(k)}(\lambda)$, $\gamma_{a,b}^{(k)}(\lambda)$, and~$\delta^{(k)}(\lambda)$ are polynomials in $\lambda$.

\begin{dfn}
The system of equations
\begin{equation} \label{e3}
Q_{2k} \psi = 0
\end{equation}
for $\psi = \psi(z, \lambda)$
is called the \emph{system of heat equations in nonholonomic frame} $L_k$. The operators $Q_{2k}$ are called \emph{Schr\"odinger operators}.
\end{dfn}

In \cite{Nonhol} a solution to the following problem is given. 

\begin{prb}\label{pro1} Find sufficient conditions on 
$\bigl\{\alpha^{(i)}(\lambda),\beta^{(i)}(\lambda),\gamma^{(i)}(\lambda),\delta^{(i)}(\lambda)\bigr\}$
for the operators \eqref{HQ} to give a representation of the Lie
algebra $\mathscr{L}_{L}$ in the ring of operators on~$\CU.$
\end{prb}

As shown in \cite{Nonhol}, the system of heat equations with Schr\"odinger operators $Q_{2k}$ that give a solution to~Problem \ref{pro1} determines the hyperelliptic sigma function $\sigma(z, \lambda)$. This allows to~construct the hyperelliptic Kleinian functions theory starting from such a system.
Further we denote by $Q_{2k}$ the Schr\"odinger operators that give a solution to~Problem~\ref{pro1}. A construction of these operators is given in \cite{Nonhol}.

In this work we show that for this system for any solution $\psi(z,\lambda)$ of \eqref{e3} such that the expression $\psi(z,0)$ is a polynomial, this polynomial coincides with the rational limit of the sigma function $\sigma(z, 0)$ up to a constant factor. The condition $\psi(1,0,\ldots,0) = 1$ normalizes this constant factor.

By Theorem 2.6 in \cite{Nonhol}, if~$\psi(z, \lambda)$ is an entire function such that $\psi(z, 0) = \sigma(z, 0)$, then it coincides with the hyperelliptic sigma function~$\sigma(z, \lambda)$. Therefore we obtain:

\begin{thm} \label{T12}
If an entire function $\psi(z, \lambda)$ is a solution of a system of heat equations~\eqref{e3} with Schr\"odinger operators $Q_{2k}$ that give a solution to~Problem \ref{pro1}, and~$\psi(z,0)$ is a polynomial with $\psi(1,0,\ldots,0) = 1$, then $\psi(z, \lambda)$ is the hyperelliptic sigma function.
\end{thm}

Let us note that for the sigma function $\sigma(z, \lambda)$ of the non-singular curve $\mathcal{V}_\lambda$ on the Jacobian variety of the curve $\mathcal{V}_\lambda$
the Abelian function
\[
\wp_{1,1}(z, \lambda) = - \frac{\partial^2}{\partial z_1^2}\ln \sigma(z, \lambda)
\]
determines a solution of the form $u(z) = 2 \wp_{1,1}(z, \lambda)$ of the Korteweg--de Vries (KdV) hierarchy, see \cite{BEL-97}. In \cite{BS} 
the inverse problem was posed and solved. Namely, in the notation of \cite{BS}, the equation 
\begin{equation}\label{eBS}
 2 \partial_x^2 \log f = - u
 \end{equation}
where $u$ is a solution of the stationary $g$-KdV equation, is considered. The problem of~supplementing \eqref{eBS} with natural conditions so that it has a unique solution is~solved. 
This problem is deeply connected with the problem of expressing the sigma function in~terms of tau functions \cite{Nak-10-AJM}. In \cite{KL} tau functions are introduced using the equation~\eqref{eBS}.

In \cite{AM} the so-called Adler--Moser polynomials were introduced. They give solutions to the stationary $g$-KdV hierarchy. 
The construction of these polynomials uses a recurrent sequence of inhomogeneous first-order differential equations with polynomial coefficients, that was introduced by Burchnall and Chaundy in \cite{B-Ch}. The remarkable property of this sequence is that it has polynomial solutions. These solutions naturally arise in a number of problems, see \cite{VW}, and are called Burchnall--Chaundy polynomials. 

The rational limit 
$\sigma(z,0)$ determines as well a solution $-2 \partial_1^2 \ln \sigma(z, 0)$ of the KdV hierarchy
(the proof uses that all the coefficients $\sigma_I(\lambda)$
of the series $\sigma(z, \lambda)$ are polynomials in $\lambda$). This naturally leads to the problem to describe the relation between the polynomial $\sigma(z,0)$ and Adler--Moser polynomials \cite{AM}. We give a solution to this problem  more precise than the result presented in \cite{Rat}.

The work is organized as follows: 

In Section \ref{S3} we give explicit formulas for the operators $H_{2k}$ that were found in~\cite{BB20}.

In~Section \ref{S4} we give the corresponding examples in the cases $g = 1,2,3$, and $4$.

In Section \ref{S5} we give all the steps to obtain Theorem \ref{T12}.
For $k = 0, 1, \ldots, 2g-1$ we~set
\[
\widehat{H}_{2k} = H_{2k}|_{\lambda = 0}.
\]
These second-order linear differential operators act on functions in $z = (z_1, \ldots, z_{2g-1})$. A direct verification shows that the operators $-\widehat{H}_{2k}$ determine 
a representation of the Lie--Witt algebra, namely, $[\widehat{H}_{2},\widehat{H}_{2k-2}] = 2(k-2)\widehat{H}_{2k}$.

The main result of Section \ref{S5} is given in Theorem \ref{ts5}:

For each genus $g$ any polynomial solution $\psi(z)$ of the system 
\begin{align} \label{H3}
\widehat{H}_0 \psi(z) &= 0, & \widehat{H}_2 \psi(z) &= 0, & \widehat{H}_4 \psi(z) &= 0
\end{align}
coincides with the rational limit of the sigma function up to a multiplicative constant.

In Section \ref{S6} we introduce differential operators  $A_{2k}$ with $k\geqslant 0$. These operators act on the ring of functions of an infinite number of variables  $z_1,z_3,\ldots,z_{2k-1},\ldots$. Note that all 
$A_{2k}$ for $k>0$ are operators of second order. Directly from the formulas for  $A_{2k},\, k\geqslant 0$,
it follows that if $z_{2s-1}=0$ and $\partial_{2s-1}=0$ for $s>g$, then $A_{2k} = - \widehat{H}_{2k}$.
We show that the Lie algebra with generators $A_{2k},\, k\geqslant 0$, over the field of rational numbers $\mathbb{Q}$
coincides with the Lie--Witt algebra $W_0$, generated by the operators $A_0,\,A_2,\,A_4$.

In Section \ref{S7} we consider the problem of constructing the Lie algebra of derivations of~$\mathcal{F}$. This Lie algebra
has~$3g$ generators $\mathcal{L}_{2k-1}$, where $k=1,\ldots,g$, and
$\mathcal{L}_{2k}$, where $k=0,\ldots,2g-1$. Here $\mathcal{L}_{2k-1} = \partial / \partial z_{2k-1}$.
The operators $\mathcal{L}_{2k}$ include as summands the differentiation operators~$L_{2k}$ over the parameters $\lambda$.
In general form, the method of~construction of the operators $\mathcal{L}_{2k}$ is given \cite{BL0}, \cite{BL}.
Here we specify the explicit form of the operators $\mathcal{L}_0,\,\mathcal{L}_2,\,\mathcal{L}_4$ for all $g\geqslant 1$.
The cases $g=1,2,3,4$ are considered in detail.

In Section \ref{S8} we describe the construction of polynomial dynamical systems that correspond to the differential operators $\mathcal{L}_{2k}$. We focus on the relation with the KdV equation. 

In Section \ref{S9} we give examples of such systems in the genus $g=3$ case.

In Section \ref{SL} we describe the connection of polynomial solutions of the system \eqref{H3} with Burchnall--Chaundy polynomials \cite{B-Ch}, \cite{VW} and Adler--Moser polynomials \cite{AM}, \cite{Rat}.

\section{Explicit expressions for Shr\"odinger operators\\ that determine hyperelliptic sigma functions}

\label{S3}

The construction of the operators $Q_{2k}$ in \cite{Nonhol} uses the condition
(see equation (1.3) in~\cite{Nonhol})
 stating that the commutator of operators $[Q_{2i},Q_{2j}]$ is determined by a formula over $P$ with the same coefficients as the formula for $[L_{2i}, L_{2j}]$, 
namely, the Lie algebra generated by the operators
$Q_{2i}$
with
$i=0, 1, ...$
is yet another realization of Lie--Witt algebra $W_0$ deformation.
For an effective description of this Lie algebra, one needs to obtain explicit formulas for $Q_0, Q_2$, and $Q_4$. These formulas were found in~\cite{BB20}. Provided $Q_{2k} = L_{2k} - H_{2k}$, we have:
\begin{align*}
H_0 &= \sum_{s=1}^g (2s-1) z_{2s-1} \partial_{2s-1} - {g (g+1) \over 2};\\
H_2 &= {1 \over 2} \partial_1^2 + \sum_{s=1}^{g-1} (2s-1) z_{2s-1} \partial_{2s+1}
- {4 \over 2 g + 1} \lambda_4  \sum_{s=1}^{g-1} (g - s)  z_{2s+1} \partial_{2s-1}
+ \\ & \qquad \qquad +\sum_{s=1}^{g} \left({2s - 1 \over 2} \lambda_{4s} - {2 (g - s + 1)  \over 2 g+1} \lambda_{4} \lambda_{4s - 4} \right) z_{2s-1}^2;
\\
H_4 &= \partial_1 \partial_3 + \sum_{s=1}^{g-2} (2s-1) z_{2s-1} \partial_{2s+3}
+ \lambda_4 \sum_{s=1}^{g-1} (2s-1) z_{2s+1} \partial_{2s+1} - \\
& - {6 \over 2 g + 1} \lambda_6 \sum_{s=1}^{g-1} (g - s) z_{2s+1} \partial_{2s-1} 
+ \sum_{s=1}^{g}\left( (2 s - 1) \lambda_{4 s +2} - {3 (g - s + 1) \over 2 g + 1} \lambda_6 \lambda_{4 s - 4} \right) z_{2 s-1}^2 +\\
& + \sum_{s=1}^{g-1} (2 s - 1) \lambda_{4s+4} z_{2s-1} z_{2s+1} - {g (g-1) \over 2} \lambda_4.
\end{align*}
Here~$\lambda_s = 0$ for all $s \notin \{0,4,6, \ldots, 4 g, 4 g + 2\}$ and $\lambda_0 = 1$.

\begin{lem}[Lemmas 3.1 and 4.2 in \cite{BB20}] \label{L21}
For Schr\"odinger operators, in~\eqref{e2} we have
\begin{align*}
&\alpha_{a,b}^{(k)}(\lambda) = 1, \quad \text{if} \quad a+b = 2 k, \quad \text{and} \quad  a, b \in 2 \mathbb{N} + 1,\\
&\alpha_{a,b}^{(k)}(\lambda) = 0, \quad \text{if} \quad a+b \ne 2 k, \quad \text{and} \quad  a, b \in 2 \mathbb{N} + 1,\\
&\beta_{a,b}^{(k)}(\lambda) \text{ is a linear function in } \lambda,\\
&\gamma_{a,b}^{(k)}(\lambda) \text{ is a quadratic function in } \lambda, \\
&\delta^{(k)}(\lambda) = \left(- {1 \over 4} (2g-k+1) (2g-k) + {1 \over 2} \left( g + \left[ {k+1 \over 2} \right] - k \right) \left( g - \left[ {k+1 \over 2} \right]\right) \right) \lambda_{2k}.
\end{align*}
\end{lem}

\begin{cor}[Corollary 4.4 in \cite{BB20}] \label{t345}
For $k = 3, 4, 5, \ldots, 2g-1$,
\[
Q_{2k} = {1 \over 2 (k-2)} [Q_2, Q_{2k-2}] - {2 (2 g - k + 1) \over (k-2) (2 g + 1)} \left( \lambda_{2k} Q_0 - \lambda_4 Q_{2k-4}\right). 
\]
\end{cor}

This relation recurrently defines the operators $Q_{2k}$ for $k = 3, 4, 5, \ldots, 2g-1$ and yields explicit expressions for these operators.

\section{Explicit formulas in the case of genus $g=1,2,3,4$} \label{S4}

In this section for illustration of the results of Section
\ref{S3} we give the  operators $H_0, H_2$ and $H_4$ for $g=1,2,3,4$. The explicit form of all the operators $H_{2k}$ for these $g$ is given in~\cite{BB20}. 

\subsection{The Schr\"odinger operators for the genus $g=1$} \text{}

In this case, the explicit formulas for $\{H_{2k}\}$ in \eqref{HQ} are
\begin{align*}
H_0 &= z_1 \partial_1 - 1; &
H_2 &= {1 \over 2} \partial_1^2 - {1 \over 6} \lambda_4 z_1^2.
\end{align*}

\subsection{The Schr\"odinger operators for the genus $g=2$} \text{}

In this case, the explicit formulas for $\{H_{2k}\}$ in \eqref{HQ} are
\begin{align*}
H_0 &= z_1 \partial_1 + 3 z_3 \partial_3 - 3;\\
H_2 &= {1 \over 2} \partial_1^2 - {4 \over 5} \lambda_4 z_3 \partial_1 + z_1 \partial_3 - {3 \over 10} \lambda_4 z_1^2 + \left({3 \over 2} \lambda_8 - {2 \over 5} \lambda_4^2\right) z_3^2;\\
H_4 &= \partial_1 \partial_3 - {6 \over 5} \lambda_6 z_3 \partial_1 + \lambda_4 z_3 \partial_3 - {1 \over 5} \lambda_6 z_1^2 + \lambda_8 z_1 z_3 + \left(3 \lambda_{10} - {3 \over 5} \lambda_4 \lambda_6\right) z_3^2 - \lambda_4.
\end{align*}

\subsection{The Schr\"odinger operators for the genus $g=3$} \text{}

In this case, the explicit formulas for $\{H_{2k}\}$ in \eqref{HQ} are 
\begin{align*}
H_0 &=
z_1 \partial_1 + 3 z_3 \partial_3 + 5 z_5 \partial_5 - 6;
\\
H_2 &= {1 \over 2} \partial_1^2 - {8 \over 7} \lambda_4 z_3 \partial_1 + \left(z_1 - {4 \over 7} \lambda_4 z_5\right) \partial_3 + 3 z_3 \partial_5 + \\
& - {5 \over 14} \lambda_4 z_1^2 + \left({3 \over 2} \lambda_8 - {4 \over 7} \lambda_4^2\right) z_3^2 +
 \left({5 \over 2} \lambda_{12} - {2 \over 7} \lambda_4 \lambda_8 \right) z_5^2;\\
H_4 &= \partial_1 \partial_3 - {12 \over 7} \lambda_6 z_3 \partial_1 + \left(\lambda_4 z_3 - {6 \over 7} \lambda_6 z_5\right) \partial_3 + \left(z_1 + 3 \lambda_4 z_5 \right) \partial_5 - {2 \over 7} \lambda_6 z_1^2 + \\
& + \lambda_8 z_1 z_3 
 + \left(3 \lambda_{10} - {6 \over 7} \lambda_4 \lambda_6\right) z_3^2 + 3 \lambda_{12} z_3 z_5 + \left(5  \lambda_{14} - {3 \over 7} \lambda_6 \lambda_8 \right) z_5^2
  - 3 \lambda_4.
\end{align*}

\subsection{The Schr\"odinger operators for the genus $g=4$} \text{}

In this case, the explicit formulas for $\{H_{2k}\}$ in \eqref{HQ} are 

\begin{align*}
H_0 &=
z_1 \partial_1 + 3 z_3 \partial_3 + 5 z_5 \partial_5 + 7 z_7 \partial_7 - 10;\\
H_2 &= {1 \over 2} \partial_1^2 + z_1 \partial_3 + 3 z_3 \partial_5 + 5 z_5 \partial_7 - {4 \over 9} \lambda_4 \left(3 z_{3} \partial_{1} + 2 z_{5} \partial_{3} + z_{7} \partial_{5} \right) - \\
& \qquad - {7 \over 18} \lambda_4 z_1^2 + \left({3 \over 2} \lambda_8 - {2 \over 3} \lambda_4^2 \right) z_3^2 + \left( {5 \over 2} \lambda_{12} - {4 \over 9} \lambda_4 \lambda_8 \right) z_5^2 + \left({7 \over 2} \lambda_{16} - {2 \over 9} \lambda_4 \lambda_{12}  \right) z_7^2;\\
H_4 &= \partial_1 \partial_3 + z_1 \partial_5 + 3 z_3 \partial_7 +
\lambda_4 \left( z_3 \partial_3 + 3 z_5 \partial_5 + 5 z_7 \partial_7 \right)
- {2 \over 3} \lambda_6 \left( 3 z_3 \partial_1 + 2 z_5 \partial_3 + z_7 \partial_5 \right)-  \\
& \qquad - {1 \over 3} \lambda_6 z_1^2 + \lambda_8 z_1 z_3 + \left(3 \lambda_{10} - \lambda_4 \lambda_6\right) z_3^2 + 3 \lambda_{12} z_3 z_5 + \left(5 \lambda_{14} - {2 \over 3} \lambda_6 \lambda_8\right) z_5^2 + \\
& \qquad + 5 \lambda_{16} z_5 z_7 + \left(7 \lambda_{18} - {1 \over 3} \lambda_6 \lambda_{12}\right) z_7^2 - 6 \lambda_4.
\end{align*}

\section{Adler--Moser polynomials and the rational limit of sigma functions} \label{S5}

Let $\psi(z, \lambda)$ be a solution of the system of~heat equations~\eqref{e3}. 
We have
\[
 L_{2k} \psi(z,\lambda) = H_{2k} \psi(z,\lambda)
\]
for $k = 0, 1, \ldots, 2g-1$. Note that from the form of the operators $L_{2k}$ it follows that for~$\lambda = 0$ we have
$L_{2k} \psi(z,0) = 0$ for all $k$. Thus, the equations
\[
H_{2k} \psi(z,0) = 0                                                                                                                                                                                                \]
hold for $k = 0, 1, \ldots, 2g-1$. 

By the \emph{rational limit} of the operators $H_{2k}$ we denote the operators
\[
\widehat{H}_{2k} = H_{2k}|_{\lambda = 0}. 
\]
From the form of the operators $H_{2k}$ it follows that 
\[
 (H_{2k}\psi)(z,0) = \widehat{H}_{2k}(\psi(z,0)).
\]
Therefore the function $\psi(z,0)$ is a solution of the system of \emph{rational limit heat equations}
\[
 \widehat{H}_{2k} \psi = 0,
\]
where $k = 0, 1, \ldots, 2g-1$. We consider its subsystem for $k = 0,1,2$.

For genus $g$ we denote by $m_g(z)$ a polynomial solution $m_g(z) = \psi(z)$ of the system 
\begin{align*}
\widehat{H}_0 \psi(z) &= 0, & \widehat{H}_2 \psi(z) &= 0, & \widehat{H}_4 \psi(z) &= 0
\end{align*}
with initial conditions $m_g(1,0,\ldots,0) = 1$.

\begin{ex} \label{ex1}
We have
$
m_1(z) = z_1.
$
\end{ex}

\begin{ex}
We have
$
m_2(z) = z_1^3 - 3 z_3.
$
\end{ex}

\begin{ex}
We have
$
m_3(z) = z_1^6 - 15 z_1^3 z_3  + 45 z_1 z_5 - 45 z_3^2.
$
\end{ex}

\begin{ex} \label{ex4}
We have
\[
m_4(z) = z_1^{10} - 45 z_1^7 z_3 + 315 z_1^5 z_5 - 1575 z_1^3 z_7 + 4725 z_1^2 z_3 z_5 - 4725 z_1 z_3^3 + 4725 z_3 z_7 - 4725 z_5^2.
\]
\end{ex}

\begin{ex} \label{ex5}
We have
\begin{multline*}
m_5(z) = z_1^{15} - 105 z_1^{12} z_3 + 1260 z_1^{10} z_5 + 1575 z_1^9 z_3^2 - 14175 z_1^8 z_7 + 14175 z_1^7 z_3 z_5 - \\ - 33075 z_1^6 z_3^3 + 99225 z_1^6 z_9 - 297675 z_1^5 z_3 z_7 - 297675 z_1^5 z_5^2 + \\ + 1488375 z_1^4 z_3^2 z_5
- 992250 z_1^3 z_3^4 - 1488375 z_1^3 z_3 z_9 + 1488375 z_1^3 z_5 z_7 + \\ + 4465125 z_1^2 z_3^2 z_7 - 4465125 z_1^2 z_3 z_5^2 - 1488375 z_1 z_3^3 z_5 + 4465125 z_1 z_5 z_9 - 4465125 z_1 z_7^2 + \\
+ 1488375 z_3^5 - 4465125 z_3^2 z_9 + 8930250 z_3 z_5 z_7 - 4465125 z_5^3. 
\end{multline*}
\end{ex}

We will further show that the polynomials $m_g(z)$ satisfy the differential equation
\[
m_{g+1}' m_{g-1} - m_{g+1} m_{g-1}' = (2g+1) m_g^2,
\]
where the prime denotes the derivative in $z_1$ and $m_0(z) = 1$.
This coincides with the description of the polynomials $\theta_k$ in \cite{AM}, that were later called Adler--Moser polynomials.
Their second logarithmic derivatives in $z_1$ give rational solutions of the Korteweg--de Vries equation. See Section \ref{SL} for more detailed definitions and results. 

\begin{rem}
In \cite{AM} the variables of the polynomials $\theta_k(\tau_1, \ldots, \tau_k)$ are determined by $\tau_1 = z_1$ and the normalization that the coefficient of $z_1^{k (k-1)/2}$ in $\theta_{k+1}$ is equal to $\tau_{k+1}$. From the Examples \ref{ex1}--\ref{ex5} we obtain for $g=1,2,3,4,5$:
\[
 m_g(z) = \theta_g(\tau_1, \ldots, \tau_g),
\]
where 
$\tau_1 = z_1$, $\tau_2 = - 3 z_3$, $\tau_3 = 45 z_5$,  $\tau_4 = - 1575 z_7$, $\tau_5 = - 33075 z_3^3 + 99225 z_9$.
\end{rem}

\begin{lem} \label{lem52} We have the rational limits
\begin{align*}
\widehat{H}_0 &= - {g (g+1) \over 2} + \sum_{s=1}^g (2s-1) z_{2s-1} \partial_{2s-1};\\
\widehat{H}_2 &= {1 \over 2} \partial_1^2 + \sum_{s=1}^{g-1} (2s-1) z_{2s-1} \partial_{2s+1};\\
\widehat{H}_4 &= \partial_1 \partial_3 + \sum_{s=1}^{g-2} (2s-1) z_{2s-1} \partial_{2s+3}.
\end{align*}
\end{lem}

This lemma is a direct corollary from the formulas of Section \ref{S3}.

From Corollary \ref{t345} for $k = 3, 4, 5, \ldots, 2g-1$ we obtain
\begin{equation} \label{eq9}
- 2 (k-2) \widehat{H}_{2k} = [\widehat{H}_2, \widehat{H}_{2k-2}].
\end{equation}

\begin{lem} \label{lem522} We have the rational limits
\[ 
\widehat{H}_{2k} = {1 \over 2} \sum_{s=1}^{k} \partial_{2s-1} \partial_{2k+1-2s} + \sum_{s=1}^{g-k} (2s-1) z_{2s-1} \partial_{2s+2k-1}
\]
for $k = 1, 2, 3, 4, 5, \ldots, 2g-1$.
Here we set $\partial_s = 0$ if $s > 2g-1$.
\end{lem}
\begin{proof}
For $k = 1$ and $k=2$ this formula holds from Lemma \ref{lem52}. For $k = 3, 4, 5, \ldots, 2g-1$ it follows from~\eqref{eq9}.
\end{proof}

This lemma leads us to the operators that we will consider in the next section of this work.

\begin{lem} \label{lem54}
The conditions
\[
 \widehat{H}_0 m_g(z) = \widehat{H}_2 m_g(z) = \widehat{H}_4 m_g(z) = 0
\]
and $m_g(1,0,\ldots,0) = 1$ determine the polynomial $m_g(z)$ uniquely.
\end{lem}

\begin{proof}
Set
\[
m_g(z) = \sum_{i_1, i_3, \ldots, i_{2g-1}}
a(i_1, i_3, \ldots, i_{2g-1}) z_1^{i_1} z_3^{i_3} \ldots z_{2g-1}^{i_{2g-1}}.
\]

The equation $\widehat{H}_0 m_g(z) = 0$ implies
\[
a(i_1, i_3, \ldots, i_{2g-1}) = 0 \qquad
\text{for} \qquad i_1 + 3 i_3 + \ldots + (2g-1) i_{2g-1} \ne {g (g+1) \over 2}.
\]

Let us find the coefficients with $ i_1 + 3 i_3 + \ldots + (2g-1) i_{2g-1} = g (g+1)/ 2$.

We have $a(g (g+1) / 2, 0, \ldots, 0) = 1$ from the condition $m_g(1,0,\ldots,0) = 1$.
We find the other coefficients $a(i_1, i_3, \ldots, i_{2g-1})$ by induction by weight 
\[
wt(i_1, i_3, \ldots, i_{2g-1}) = i_1 + i_3 + \ldots + i_{2g-1} 
\]
starting from the highest weight to the lowest.
Among the coefficients of the same weight we find first the coefficients with $i_{2g-1} \ne 0$, then with $i_{2g-1} = 0, i_{2g-3} \ne 0$, and so on, with the coefficients $a(i_1, i_3, \ldots, i_{2g-2k+1}, 0, \ldots, 0)$, where $i_{2g-2k+1} \ne 0$, on the $k$-th step.

Consider the operator (see Lemma \ref{lem522}):
\[
\widehat{H}_{2(g-k)} = {1 \over 2} \sum_{s=1}^{g-k} \partial_{2s-1} \partial_{2g-2k-2s+1}  +  z_{1} \partial_{2g-2k+1} + \sum_{s=2}^{k} (2s-1) z_{2s-1} \partial_{2g-2k+2s-1}.
\]
By \eqref{eq9} we have $\widehat{H}_{2(g-k)} m_g(z) = 0$. The coefficient at
\[
z_1^{i_1+1} z_3^{i_3} \ldots z_{2g-2k-1}^{ i_{2g-2k-1}} z_{2g-2k+1}^{i_{2g-2k+1}-1} 
\]
in this equation 
gives a relation between 
$a(i_1, i_3, \ldots, i_{2g-2k+1}, 0, \ldots, 0)$ and the coefficients $a(j_1, j_3, \ldots, j_{2g-1})$, where either 
\[
wt(j_1, j_3, \ldots, j_{2g-1}) = wt(i_1, i_3, \ldots, i_{2g-2k+1}, 0, \ldots, 0) + 2, 
\]
or $j_p = 1$ for some $p > 2g-2k+1$. In fact, it gives an expression for 
\[
i_{2g-2k+1} a(i_1, i_3, \ldots, i_{2g-2k+1}, 0, \ldots, 0) 
\]
as a polynomial with integer coefficients in $a(j_1, j_3, \ldots, j_{2g-1})$, where the conditions on~$(j_1, j_3, \ldots, j_{2g-1})$ are given above. This provides the step of the induction and thus proves the statement of the Lemma.
\end{proof}

\begin{cor} \label{cor55}
 There is no non-zero polynomial solution $m_g^0(z) = \psi(z)$ to the system
\begin{align*}
\widehat{H}_0 \psi(z) &= 0, & \widehat{H}_2 \psi(z) &= 0, & \widehat{H}_4 \psi(z) &= 0
\end{align*}
with the condition $m_g^0(1,0,\ldots,0) = 0$.
\end{cor}

\begin{proof}
For a solution $m_g(z)$ of Lemma \ref{lem54} the expression $m_g(z) + m_g^0(z)$ gives another solution as the problem is linear. This contradicts the statement of the Lemma.
\end{proof}

\begin{cor} \label{cor56}
The conditions
\begin{align*}
\widehat{H}_0 \psi(z) &= 0, & \widehat{H}_2 \psi(z) &= 0, & \widehat{H}_4 \psi(z) &= 0
\end{align*}
determine the polynomial $m_g^c(z) = \psi(z)$ uniquely up to a multiplicative constant.
\end{cor}

\begin{proof}
Let $m_g^c(1,0,\ldots,0) = c$, by Corollary \ref{cor55} we have $c \ne 0$ for a non-zero polynomial~$m_g^c(z)$, thus
we have $m_g^c(z) = c m_g(z)$ for the solution $m_g(z)$ of Lemma \ref{lem54}.
\end{proof}

\begin{thm} \label{ts5} 
For each genus $g$ any polynomial solution $\psi(z)$ of the system 
\begin{align} \label{s11}
\widehat{H}_0 \psi(z) &= 0, & \widehat{H}_2 \psi(z) &= 0, & \widehat{H}_4 \psi(z) &= 0
\end{align}
coincides with the rational limit of the sigma function up to a multiplicative constant.
\end{thm}

\begin{proof}
By Theorem 2.6 in \cite{Nonhol} the function $\sigma(z, \lambda)$ is a solution of the system of heat equations \eqref{e3}. Thus, $\sigma(z, 0)$ is a polynomial solution of the system~\eqref{s11}, and by Corollary~\ref{cor56} it coincides with any other polynomial solution up to a multiplicative constant.
\end{proof}

\begin{thm}
For each genus $g$ any polynomial solution $\psi(z)$ of the system 
\begin{align*} 
\widehat{H}_0 \psi(z) &= 0, & \widehat{H}_2 \psi(z) &= 0, & \widehat{H}_4 \psi(z) &= 0
\end{align*}
coincides with the corresponding Adler--Moser polynomial up to a multiplicative constant and a change of variables.
\end{thm}

\begin{proof}
This result follows from the relation between Adler--Moser polynomials, Shur--Weierstrass polynomials and rational limits of sigma functions. These results will be discussed in detail in Section \ref{SL}. See also Theorem 6.3 in \cite{Rat} and  Exercise 7.35(c) in~\cite{Stanley}.
\end{proof}

\section{Lie subalgebra of Witt algebra} \label{S6}

For $k = 0,1,2,\ldots$ we set 
\begin{equation} \label{A2k}
A_{2k} = - {1 \over 2} \sum_{s=1}^{k} \partial_{2s-1} \partial_{2k+1-2s} - \sum_{s=1}^{\infty} (2s-1) z_{2s-1} \partial_{2s+2k-1}.
\end{equation}
Note that $A_0$ is the Euler vector field:
\[
A_{0} = - \sum_{s=1}^{\infty} (2s-1) z_{2s-1} \partial_{2s-1}.
\]

Given~$g$, we set $z_{2s-1} \equiv 0$ for $s > g$, thus for each $g$ the sum in \eqref{A2k} is finite. If we set $\partial_s = 0$ for $s > 2g-1$, we obtain the relations with the rational limits of the operators~$H_{2k}$:
\[
A_{2k} = - \widehat{H}_{2k}
\]
for $k = 1,2,\ldots,2g-1$,
and 
\[
A_{0} = - \widehat{H}_{0} - {g (g+1) \over 2}.
\]

The following Lemma shows that the operators $A_{2k}$ are the generators of the Lie subalgebra $W_0$ of the Witt algebra. In \cite{BM} the Lie subalgebra $W_{-1}$ of the Witt algebra appears in a closely related construction.

\vfill

\eject

\begin{lem} \label{lem61} The commutation relation holds:
\[
[A_{2i}, A_{2j}] = 2 (j-i) A_{2(i+j)}.
\]
\end{lem}
\begin{proof}
We have 
\begin{multline*}
\sum_{s=1}^i (2s -1) \partial_{2i+1-2s} \partial_{2s + 2j - 1} - \sum_{s=1}^j (2s -1) \partial_{2j+1-2s} \partial_{2s + 2i - 1} = \\
= \sum_{s=1}^i (2i - 2s +1) \partial_{2s-1} \partial_{2i+2j-2s+1} - \sum_{s=1}^j (2j - 2s +1) \partial_{2s-1} \partial_{2i+2j-2s+1} = 
\\
=
2 (i - j) \sum_{s=1}^{min(i,j)} \partial_{2s-1} \partial_{2i+2j-2s+1} + \\
+ sign(i-j) \sum_{s=min(i,j)+1}^{max(i,j)} (2 \, max(i,j)-2s+1) \partial_{2s-1} \partial_{2i+2j-2s+1}
= \end{multline*}
\begin{multline*}
=
(i - j) \sum_{s=1}^{min(i,j)} \partial_{2s-1} \partial_{2i+2j-2s+1} + (i - j) \sum_{s=max(i,j)+1}^{i+j} \partial_{2s-1} \partial_{2i+2j-2s+1} + \\
+ {1 \over 2} sign(i-j) \sum_{s=min(i,j)+1}^{max(i,j)} (2 \, max(i,j)-2s+1) \partial_{2s-1} \partial_{2i+2j-2s+1} + \\
+ {1 \over 2} sign(i-j) \sum_{s=min(i,j)+1}^{max(i,j)} (2s-1-2 \, min(i,j)) \partial_{2s-1} \partial_{2i+2j-2s+1}
= \\
= (i-j)  \sum_{s=1}^{i+j} \partial_{2s-1} \partial_{2i+2j+1-2s};
\end{multline*}
therefore,
\begin{multline*}
[A_{2i}, A_{2j}] = \sum_{s=1}^i (2s -1) \partial_{2i+1-2s} \partial_{2s + 2j - 1} - \sum_{s=1}^j (2s -1) \partial_{2j+1-2s} \partial_{2s + 2i - 1} + \\
+ \sum_{s=1}^\infty (2s +2i-1) (2s-1) z_{2s-1} \partial_{2s+2i+2j-1} 
- \sum_{s=1}^\infty (2s +2j-1) (2s-1) z_{2s-1} \partial_{2s+2i+2j-1} = \\ = (i-j)  \sum_{s=1}^{i+j} \partial_{2s-1} \partial_{2i+2j+1-2s} + 2 (i-j) \sum_{s=1}^\infty (2s-1) z_{2s-1} \partial_{2s+2i+2j-1} = 2 (j-i) A_{2(i+j)}.
\end{multline*}
\vspace{-5mm}

\end{proof}

\vfill

\eject

\section{Derivations of the field of genus $g$ hyperelliptic functions} \label{S7}

In this section we give explicitly a part of the solution to the problem of constructing the Lie algebra of derivations of $\mathcal{F}$, i.e. of~finding~$3g$ independent differential operators~$\mathcal{L}$ such that $\mathcal{L} \mathcal{F} \subset \mathcal{F}$. The setting of the problem, as well as a general approach to the solution, can be found in \cite{BL0} and \cite{BL}. An~overview is given in \cite{BEL18}. In~\cite{FS}, \cite{B2}, \cite{B3}, and~\cite{B4} an explicit solution to this problem was obtained for $g=1,2,3,4$.

For any $g$ the operators $\mathcal{L}_{2k-1} = \partial_{2k-1}$, $k \in \{1,2,\ldots,g\}$, belong to the set of generators of this Lie algebra.
Here we give explicitly three of its other generators. 
\begin{thm}
The operators \vspace{-5pt}
\begin{align*}
\mathcal{L}_0 &= L_{0} - \sum_{s=1}^g (2s-1) z_{2s-1} \partial_{2s-1}; \\
\mathcal{L}_2 &= L_{2} - \zeta_1 \partial_1 - \sum_{s=1}^{g-1} (2s-1) z_{2s-1} \partial_{2s+1}
+ {4 \over 2 g + 1} \lambda_4  \sum_{s=1}^{g-1} (g - s)  z_{2s+1} \partial_{2s-1};\\
\mathcal{L}_4 &= L_{4} - \zeta_3 \partial_1 - \zeta_1 \partial_3 - \sum_{s=1}^{g-2} (2s-1) z_{2s-1} \partial_{2s+3}
- \\
& \qquad \qquad - \lambda_4 \sum_{s=1}^{g-1} (2s-1) z_{2s+1} \partial_{2s+1} + {6 \over 2 g + 1} \lambda_6 \sum_{s=1}^{g-1} (g - s) z_{2s+1} \partial_{2s-1}
\end{align*}
belong to the Lie algebra of derivations of $\mathcal{F}$. \vspace{-4pt}
\end{thm}

\begin{proof}
The Theorem follows directly from the explicit formulas given  in Section \ref{S3} and the Theorems 13 and 14 of \cite{BL}.
\end{proof}

Let us give for illustration the formulas for the corresponding operators for $g=1,2,3,4$ (cf. \cite{FS}, \cite{B2}, \cite{B3}, and \cite{B4}). They follow from formulas of Section \ref{S4}.

\subsection{Differential operators for genus $g=1$} \text{}
\begin{align*}
\mathcal{L}_0 &= L_0 - z_1 \partial_1; &
\mathcal{L}_2 &= L_2 - \zeta_1 \partial_1.
\end{align*}

\subsection{Differential operators for genus $g=2$} \text{}
\begin{align*}
\mathcal{L}_0 &= L_0 - z_1 \partial_1 - 3 z_3 \partial_3; &
\mathcal{L}_2 &= L_2 - \zeta_1 \partial_1 + {4 \over 5} \lambda_4 z_3 \partial_1 - z_1 \partial_3;\\ & & 
\mathcal{L}_4 &= L_4 - \zeta_3 \partial_1 - \zeta_1  \partial_3 + {6 \over 5} \lambda_6 z_3 \partial_1 - \lambda_4 z_3 \partial_3.
\end{align*}

\subsection{Differential operators for genus $g=3$} \text{}
\begin{align*}
\mathcal{L}_0 &= L_0 - 
z_1 \partial_1 - 3 z_3 \partial_3 - 5 z_5 \partial_5;
\\
\mathcal{L}_2 &= L_2 - \zeta_1 \partial_1 + {8 \over 7} \lambda_4 z_3 \partial_1 - z_1 \partial_3 + {4 \over 7} \lambda_4 z_5 \partial_3 - 3 z_3 \partial_5;\\
\mathcal{L}_4 &= L_4 - \zeta_3 \partial_1 - \zeta_1  \partial_3 + {12 \over 7} \lambda_6 z_3 \partial_1 - \lambda_4 z_3 \partial_3 + {6 \over 7} \lambda_6 z_5 \partial_3 - z_1 \partial_5 - 3 \lambda_4 z_5 \partial_5.
\end{align*}

\subsection{Differential operators for genus $g=4$} \text{}
\begin{align*}
\mathcal{L}_0 &= L_0 - 
z_1 \partial_1 - 3 z_3 \partial_3 - 5 z_5 \partial_5 - 7 z_7 \partial_7;\\
\mathcal{L}_2 &= L_2 - \zeta_1 \partial_1 - z_1 \partial_3 - 3 z_3 \partial_5 - 5 z_5 \partial_7 + {4 \over 9} \lambda_4 \left(3 z_{3} \partial_{1} + 2 z_{5} \partial_{3} + z_{7} \partial_{5} \right);\\
\mathcal{L}_4 &= L_4 - \zeta_3 \partial_1 - \zeta_1 \partial_3 - z_1 \partial_5 - 3 z_3 \partial_7 -
\lambda_4 \left( z_3 \partial_3 + 3 z_5 \partial_5 + 5 z_7 \partial_7 \right)
+ {2 \over 3} \lambda_6 \left( 3 z_3 \partial_1 + 2 z_5 \partial_3 + z_7 \partial_5 \right). 
\end{align*}

\vfill

\eject

\section{Polynomial dynamical systems and Korteweg--de Vries equation} \label{S8}

Consider the complex linear space $\mathbb{C}^{3g}$ with coordinates $\{ x_{i,j} \}$, where $i \in \{1,2,3\}$, $j \in \{1,3,\ldots,2g-1\}$. In the notation of Section \ref{S0}, we define the map $\mathbb{C}^{g} \times \mathcal{B} \to \mathbb{C}^{3g}$ by the relations 
\begin{equation} \label{star}
\begin{pmatrix}
x_{1,j} & x_{2,j} & x_{3,j}
 \end{pmatrix} =
 \begin{pmatrix}
\wp_{1,j}(z,\lambda) & \wp_{1,1,j}(z,\lambda) & \wp_{1,1,1,j}(z,\lambda)
 \end{pmatrix}
\end{equation}
for all $j \in \{1,3,\ldots,2g-1\}$. By this map, rational functions in $\mathbb{C}^{3g}$ correspond to genus~$g$ hyperelliptic functions (see Section 5 in  \cite{B3} for details). The map gives a correspondence of differential operators $\mathcal{L}_0, \mathcal{L}_2, \mathcal{L}_4$, and $\mathcal{L}_{2k-1}$, $k \in \{1,2,\ldots, g \}$ from Section \ref{S7} with polynomial vector fields $\mathcal{D}_0, \mathcal{D}_2, \mathcal{D}_4$, and $\mathcal{D}_{2k-1}$, $k \in \{1,2,\ldots, g \}$ in $\mathbb{C}^{3g}$. In the cases $g = 1,2,3,4$ these vector fields are given explicitly in \cite{B2}, \cite{B3}, and \cite{B4}.

For any genus $g$ we have (see equation (22) in \cite{B3})
\begin{align*} 
\mathcal{D}_0 &= \sum_j (j+1) x_{1,j} {\partial \over \partial x_{1,j}} + (j+2) x_{2,j} {\partial \over \partial x_{2,j}} +
(j+3) x_{3, j} {\partial \over \partial x_{3,j}} 
\end{align*}
and (see equation (23) in \cite{B3})
\begin{align*}
\mathcal{D}_1 &= \sum_j x_{2,j} {\partial \over \partial x_{1,j}} + x_{3, j} {\partial \over \partial x_{2,j}} + 4 (2 x_{1,1} x_{2,j} + x_{2,1} x_{1, j} + x_{2, j+2}) {\partial \over \partial x_{3,j}},
\end{align*}
where $x_{2,2 g + 1} = 0$. 

Let us describe the graded homogeneous dynamical systems in $\mathbb{C}^{3g}$ determined by these vector fields. Here we follow the approach of \cite{B2}, where such a description is given in the case of the genus~$g=2$.

The dynamical system $S_0$, corresponding to the Euler vector field $\mathcal{D}_0$, has the form
\[
 {\partial \over \partial \tau_0} x_{i,j} = (i+j) x_{i,j}, \quad i = 1,2,3, \quad j = 1,3,\ldots,2g-1.
\]
The dynamical system $S_1$, corresponding to the vector field $\mathcal{D}_1$, has the form
\begin{align} \label{eq19}
{\partial \over \partial \tau_1} x_{i,j} &= x_{i+1,j}, \quad i = 1,2, \quad j = 1,3,\ldots,2g-1,\\
{\partial \over \partial \tau_1} x_{3,j} &= 4 (2 x_{1,1} x_{2,j} + x_{2,1} x_{1, j} + x_{2, j+2}), \quad j = 1,3,\ldots,2g-1, 
\end{align}
where $x_{2,2 g + 1} = 0$. 
The equality \eqref{eq19} implies
\[
 {\partial \over \partial \tau_1} x_{3,1} = 4 (3 x_{1,1} x_{2,1} + x_{2,3}).
\]
The map \eqref{star} brings $\partial / \partial \tau_1$ into $\partial_1$ and $x_{i,j}$ into the corresponding hyperelliptic $\wp$-functions. This gives the KdV equation:
\[
4 \partial_3 \wp_{1,1} = \partial_1 (\partial_1^2 \wp_{1,1} - 6 \wp_{1,1}^2).
\]
The same result can be obtained directly from the algebraic relations between hyperelliptic $\wp$-functions, see Corollary 8 in \cite{BMulti}.

Using the results of Section \ref{S5}, we see that in the rational limit $\lambda = 0$ this gives the KdV equation 
\[
4 \partial_3 \widehat{\wp} = \partial_1 (\partial_1^2 \widehat{\wp} - 6 \widehat{\wp}^2)
\]
for the function $\widehat{\wp} = - \partial_1 \partial_1 \ln \widehat{\sigma}$,
where $\widehat{\sigma}$ is the rational limit of the sigma function.
 
\vfill

\eject

\section{Examples: Polynomial dynamical systems corresponding to~derivations of hyperelliptic functions of genus $g=3$} \label{S9}

The dynamical system $S_0$, corresponding to the Euler vector field $\mathcal{D}_0$, has the form
\[
 {\partial \over \partial \tau_0} x_{i,j} = (i+j) x_{i,j}, \quad i = 1,2,3, \quad j = 1,3,5.
\]
The dynamical system $S_1$, corresponding to the vector field $\mathcal{D}_1$, has the form 
\begin{align*}
{\partial \over \partial \tau_1} x_{i,j} &= x_{i+1,j}, \quad i = 1,2, \quad j = 1,3,5,\\
{\partial \over \partial \tau_1} x_{3,j} &= 4 (2 x_{1,1} x_{2,j} + x_{2,1} x_{1, j} + x_{2, j+2}), \quad j = 1,3,5, 
\end{align*}
where $x_{2,7} = 0$. 

The dynamical system $S_2$, corresponding to the vector field $\mathcal{D}_2$, has the form
\begin{align*}
{\partial \over \partial \tau_2} x_{1,1} &= {12 \over 7} \lambda_{4} + 2 x_{1,1}^2 + 4 x_{1,3}, \\
{\partial \over \partial \tau_2} x_{2,1} &= 3 x_{1,1} x_{2,1} + 5 x_{2,3}, \\
{\partial \over \partial \tau_2} x_{3,1} &= 3 x_{2,1}^2 + 2 x_{1,1} x_{3,1} + 6 x_{3,3}, \\
{\partial \over \partial \tau_2} x_{1,3} &= - {8 \over 7} \lambda_4 x_{1,1} + 2 x_{1,1} x_{1,3} + 6 x_{1,5}, \\
{\partial \over \partial \tau_2} x_{2,3} &= - {8 \over 7} \lambda_4 x_{2,1} + 3 x_{2,1} x_{1,3} + 7 x_{2,5}, \\
{\partial \over \partial \tau_2} x_{3,3} &= - {8 \over 7} \lambda_4 x_{3,1} + 4 x_{3,1} x_{1,3} + 3 x_{2,1} x_{2,3} - 2 x_{1,1} x_{3,3} + 8 x_{3,5}, \hspace{-25mm}  \\
{\partial \over \partial \tau_2} x_{1,5} &= - {4 \over 7} \lambda_4 x_{1,3} + 2 x_{1,1} x_{1,5},\\
{\partial \over \partial \tau_2} x_{2,5} &= - {4 \over 7} \lambda_4 x_{2,3} + 3 x_{2,1} x_{1,5},\\
{\partial \over \partial \tau_2} x_{3,5} &= - {4 \over 7} \lambda_4 x_{3,3} + 4 x_{3,1} x_{1,5} + 3 x_{2,1} x_{2,5} - 2 x_{1,1} x_{3,5}, \hspace{-25mm} & &
\end{align*}
where
\begin{align*}
\lambda_{4} &= - 3 x_{1,1}^2 + {1 \over 2} x_{3,1} - 2 x_{1,3}.
\end{align*}
The calculation of this system is based on the results of \cite{B3}.

\vfill

\eject

\section{Addendum. Hyperelliptic
Schur--Weierstrass polynomials, \\ Adler--Moser polynomials and the rational limit of sigma functions
} \label{SL}

We denote by $\be = (e_1, e_2, \ldots)$ the infinite vector with coordinates $e_k$. For convenience, we assume $e_0 = 1$ and $e_k = 0$ for $k<0$. We introduce the $(g\times g)$-matrices
$
\mathcal{E}_g = (e_{g-2i+j+1})$, where $1\leqslant i,j \leqslant g$.
We set $\wt(e_k) = - k$. 

\begin{dfn}
The \emph{genus $g$ hyperelliptic Schur polynomial} is the polynomial 
\[
{\Sh}_g(\be) = \det \mathcal{E}_g.
\]
\end{dfn}
It is a homogeneous polynomial of weight $\frac{1}{2}g(g+1)$ in the $2g-1$ variables~$(e_1,\ldots,e_{2g-1})$.

\begin{ex} $\mathcal{E}_2 =
\begin{pmatrix}
e_2 & e_3\\
1 & e_1
\end{pmatrix};\quad
{\Sh}_2(\be) = e_2e_1-e_3$.
\end{ex}

\begin{ex} $\mathcal{E}_3 =
\begin{pmatrix}
e_3 & e_4 & e_5 \\
e_1 & e_2 & e_3 \\
0 & 1 & e_1
\end{pmatrix};\quad
{\Sh}_3(\be) = e_3e_2e_1+e_5e_1-e_4e_1^2-e_3^2$.
\end{ex}

Now we consider each coordinate $e_k$ as 
the $k$-th elementary symmetric polynomial in~an~infinite number of variables $\bx = (x_1, x_2, \ldots)$.
We set $\wt(x_i) = - 1$. 
The generating series for $e_k(\bx)$ is
\[
E(\bx,t) = 1+\sum_{k>0} e_k(\bx)t^k = \prod_{i>0} (1+x_it).
\]
We denote by $\Sym$ the ring of symmetric polynomials in $\bx$
over the field of rational numbers~$\mathbb{Q}$. There is an isomorphism of graded rings $\Sym \cong \mathbb{Q}[e_1, e_2, \ldots]$.
Another multiplicative homogeneous basis in $\Sym$ is given by Newton polynomials $p_1, p_2, \ldots, p_k, \ldots$, where
\[
p_k = \sum_{i>0}x_i^k.
\]
We denote by $\bp = (p_1, p_2, \ldots)$ the infinite vector with coordinates $p_k$.
The transition from the basis $\{ e_k \}$ to the basis $\{ p_k \}$ is described by the following relation
between the~generating series
\begin{equation}\label{F-4}
E(\bx,t) = \exp \mathcal{N}(\bx,t),
\end{equation}
where $\mathcal{N}(\bx,t) = \sum\limits_{n>0} (-1)^{n-1} p_n(\bx)\frac{t^n}{n}$. We have $\wt(p_k) = - k$.

One can express $e_k(\bx)$ as a polynomial in $p_1,\ldots,p_k$.
It follows from the formula \eqref{F-4} that
\begin{equation}\label{F-5}
e_k(\bx) = \frac{1}{k!}p_1^k+\ldots+(-1)^{k-1}\frac{1}{k}p_k, \quad \text{and} \quad
\frac{\partial e_k(\bx)}{\partial p_n} = (-1)^{n-1}\frac{1}{n}e_{k-n}(\bx).
\end{equation}

\begin{ex} $e_2(\bx) = \frac{1}{2}(p_1^2-p_2)$.
\end{ex}

\begin{ex} $e_3(\bx) = \frac{1}{3!}(p_1^3-3p_2p_1+2p_3)$.
\end{ex}

The general Schur--Weierstrass polynomials were introduced in \cite{Rat}. They correspond to a more general case of $(n,s)$-curves than the hyperelliptic curves case considered here.

\begin{dfn}
The Schur polynomial ${\Sh}_g(\be)$, written as a polynomial in multiplicative generators $p_1, p_2, \ldots$,
is called the \emph{genus
$g$ hyperelliptic Schur--Weierstrass polynomial} ${\ShW}_g(\bp)$. 
\end{dfn}

\begin{lem}[see \cite{Rat}]
The Schur--Weierstrass polynomial ${\ShW}_g(\bp)$ is a polynomial in $g$ variables $p_{2k-1},\, k = 1,\ldots,g$.
\end{lem}

\begin{ex} ${\ShW}_2(\bp) = \frac{1}{3}(p_1^3-p_3)$.
\end{ex}

\begin{ex} ${\ShW}_3(\bp) = \frac{1}{45}(p_1^6-5p_1^3p_3+9p_1p_5-5p_3^2)$.
\end{ex}

Set ${\ShW}_0(\bp) = 1$. We have ${\ShW}_1(\bp) = p_1$. To describe the properties of a sequence of polynomials 
${\ShW}_g(\bp),\,g=0,1,2,\ldots,$ we will need some results about Wronskians.
We will mainly follow the notation of \cite{AM}.
Consider a sequence $\{ \psi_k(x)\}$, where $k=0,1,2,\ldots$, of smooth functions in $x$, with $\psi_0(x) = 1$
and $\psi_1(x) = x$.
We set $D = \frac{d}{dx}$.
The Wronskian of a set of functions $\psi_1,\psi_2,\ldots,\psi_k$ is the determinant  $W_k = W_k(\psi_1,\psi_2,\ldots,\psi_k) = \det(D^{i-1}\psi_j),\,1\leqslant i,j \leqslant k$, of the Wronski matrix $(D^{i-1}\psi_j)$.
We obtain a sequence of smooth functions  $W_k = W_k(x)$, $k=1,2,\ldots,$ where $W_1 = x$. We set $W_0 = 1$ and $f' = Df$.

\begin{lem}[\cite{AM}, Equation (2.23)]\label{L-4}
Let $\psi_j'' = \psi_{j-1}$. Then the Wronskian sequence $W_k(\psi_1,\psi_2,\ldots,\psi_k)$,\, $k=0,1,2,\ldots,$
satisfies the system of functional differential equations
\[
W_{k+1}' W_{k-1} - W_{k+1} W_{k-1}' = W_k^2.
\]
\end{lem}

\begin{dfn}
By \emph{Burchnall--Chaundy equations} we denote a system of functional differential equations 
\begin{equation}\label{F-8}
\varphi_{k+1}'\varphi_{k-1} - \varphi_{k+1}\varphi_{k-1}' = \varphi_k^2
\end{equation}
with initial conditions $\varphi_0 = 1$ and $\varphi_1 = x$.
\end{dfn}

Burchnall and Chaundy showed in \cite{B-Ch} that the sequence of smooth functions
$\varphi_k(x),\, k=0,1,2,\ldots,$ 
that gives a solution of the system \eqref{F-8} with initial conditions $\varphi_0 = 1$ and $\varphi_1 = x$
is \emph{polynomial} and each polynomial $\varphi_{k}$, where $k>1$, has $k-1$ free parameters. These equations
arise naturally in various problems, for example, see \cite{VW}.
In works on this topic the polynomials $\varphi_k(x)$ are called Burchnall--Chaundy polynomials. 

\begin{lem}
The sequence of Schur--Weierstrass polynomials ${\ShW}_g(\bp)$, considered as a sequence of polynomials 
$\varphi_g(x)$, where $x = p_1$ and $p_3,\ldots,p_{2g-1}$ are fixed, is the sequence of Burchnall--Chaundy polynomials. 
\end{lem}
\begin{proof}
Consider the polynomials $e_k(\bx),\, k=0,1,2,\ldots,$ as polynomials in $x = p_1$
with fixed $p_2,\ldots,p_{k}$. According to the formula \eqref{F-5} we get $e_k' = e_{k-1}$. Therefore, the matrix $\mathcal{E}_g$ can be written as $(D^{i-1}e_{g-i+j}(\bx))$.
Thus, the Schur--Weierstrass polynomial  ${\ShW}_g(\bp)$ coincides with the Wronskian $W_g(e_{2g-1}(\bx),\ldots,e_g(\bx))$.
Since $e_k'' = e_{k-2}$, by Lemma \ref{L-4} this completes the proof.
\end{proof}

\begin{dfn}
A sequence of polynomials $\theta_k(x),\, k=0,1,2,\ldots$, where $\theta_0(x) = 1$ and $\theta_1(x) = x = \tau_1$, satisfying the system of functional differential equations 
\begin{equation}\label{F-9}
\theta_{k+1}'\theta_{k-1} - \theta_{k+1}\theta_{k-1}' = (2k+1)\theta_k^2,
\end{equation}
is called an \emph{Adler--Moser sequence}. 
\end{dfn}

Set $\mu_0 = \mu_1 = 1$ and
\[
\mu_k = 3^{k-1}\cdot 5^{k-2}\cdots(2k-1) = \prod_{j=1}^k (2k-2j+1)^j,\; k>1.
\]
A direct verification shows that 
\begin{lem}\label{L-44}
Let $\varphi_k(x)$, $k=0,1,2,\ldots,$ be a solution of the system of equations  \eqref{F-8} with initial conditions $\varphi_0 = 1$
and $\varphi_1 = x$.Then the sequence of functions $\theta_k(x) = \mu_k\varphi_k(x)$ determines the solution of the system  \eqref{F-9}
with initial conditions $\theta_0 = 1$ and $\theta_1 = x$.
\end{lem}

In \cite{AM} it is shown that the general solution of the system of equations  \eqref{F-9} has the form
$\theta_k = \theta_k(x,\tau_2,\ldots,\tau_k)$, where $\tau_2,\ldots,\tau_k$ are free parameters not depending on~$x$.

From equations \eqref{F-9} it follows that the choice of parameters $\tau_k,\, k\geqslant 2$, becomes unambiguous
if we fix the normalization
\begin{align}
\label{F-10}
\theta_k(x) =&\, x^{\frac{k(k+1)}{2}} + \ldots,\; k\geqslant 0, \\[4pt]
\frac{\partial\theta_k(x,\tau_2,\ldots,\tau_k)}{\partial\tau_k} =&\, x^{\frac{(k-2)(k-1)}{2}} + \ldots,\; k\geqslant 2. \label{F-11}
\end{align}

This allows us to introduce the following notion.

\begin{dfn}\label{D-4}
\emph{Universal Adler--Moser polynomials} are the polynomials $\theta_k(\tau_1,\ldots,\tau_k)$ in $k$ independent variables, where $k=0,1,2,\ldots$, 
such that in $x = \tau_1$ they give the \emph{general solution} of the system of equations \eqref{F-9} and satisfy
\eqref{F-10} and \eqref{F-11}.
\end{dfn}

The polynomials introduced by Adler and Moser in \cite{AM} are a specialization of the universal Adler--Moser polynomials from definition \ref{D-4}.

Using the formula \eqref{F-5} for $k\geqslant 2$, we obtain the equality
\[
\frac{\partial {\ShW}_k(\bp)}{\partial p_{2k-1}} = (-1)^{k+1}{\ShW}_{k-2}(\bp).
\]

\begin{thm}\label{T-1}
The Schur--Weierstrass polynomials  ${\ShW}_k(\bp)$ are uniquely determined by the fact that they satisfy the system of equations 
\eqref{F-8} and the initial conditions 
\begin{align*}
{\ShW}_k(\bp) =&\, \mu_k p_1^{\frac{k(k+1)}{2}} + \ldots,\; k\geqslant 0, \\[4pt]
\frac{\partial{\ShW}_k(\bp)}{\partial p_{2k-1}} =&\, (-1)^{k+1} \frac{1}{2k-1} \mu_{k-2}\, p_1^{\frac{(k-2)(k-1)}{2}} + \ldots,\; k\geqslant 2.
\end{align*}
\end{thm}
\begin{proof}
Using Lemma \ref{L-44} and the Definition  \ref{D-4}, we obtain the statement of the theorem.
\end{proof}

\begin{cor}
The formula holds
\[
{\ShW}_g(\bp) = \frac{1}{\mu_g}\, \theta_g(\tau_1,\ldots,\tau_g),\; \text { where }\;
\tau_k = (-1)^{k+1} \frac{1}{2k-1}\, \frac{\mu_{k}}{\mu_{k-2}}\, p_{2k-1}.
\]
\end{cor}

\begin{ex} $\tau_2 = -p_3,\quad \tau_3 = 9p_5$.
\end{ex}

Set $\widehat\tau = (\tau_2,\ldots,\tau_g),\; g\geqslant 2$, and
\[
u_g(x,\widehat\tau) = -\frac{\partial^2}{\partial x^2}\ln \theta_g(x,\widehat\tau).
\]
\begin{thm}[compare with \cite{AM}]\label{T-2}
There is a uniquely defined change of variables $\widehat\tau \to \widehat\tau^*\,:\, \tau_k = b_k \widehat\tau_k^* + h_{2k-1}(\tau^*)$,
where $b_k \in \mathbb{Q}$ and $h_{2k-1}(\tau^*)$ is a homogeneous polynomial of weight $2k-1$, such that the function $u_g(x,\widehat\tau^*)$
satisfies the KdV hierarchy corresponding to the operator  $L_g = \frac{\partial^2}{\partial x^2} + u_g(x,\widehat\tau^*)$.
\end{thm}

\begin{thm}[see \cite{ACV}] \label{T-3}
The change of variables $\widehat\tau \to \widehat\tau^*$ is described by the following relation between the generating series 
\begin{equation} \label{f-9}
\sum_{i\geqslant 2}\frac{\tau_i}{\alpha_{2i-1}} t^{2i-1} = {\rm th}\left( \sum_{i\geqslant 2}\tau_i^*t^{2i-1} \right),
\end{equation}
where $\alpha_{2i-1} = (-1)^{i-1} 3^2 \cdots(2i-3)^2(2i-1)$.
\end{thm}

The hyperbolic tangent ${\rm th}(t)$ is the exponential of the formal group over the ring of integers $\mathbb{Z}$ with the addition law  
\[
t_1\oplus^{\rm th}t_2 = \frac{t_1+t_2}{1+t_1 t_2}.
\]
Therefore, the formula \eqref{f-9} can be rewritten as
\[
\sum_{i\geqslant 2} \frac{\tau_i}{\alpha_{2i-1}}t^{2i-1} = \underset{i\geqslant 2}{\oplus}^{\rm th}\, Z_{2i-1}(t),
\]
where $Z_{2i-1}(t) = {\rm th}(\tau_i^* t^{2i-1})$.

Set $\widehat p = (p_3,\ldots,p_{2g-1})$, $\widehat z = (z_3,\ldots,z_{2g-1})$, and
denote by $\widehat{\sigma}(z_1,\widehat z)$ the rational limit of the genus $g$ hyperelliptic sigma function.

\begin{thm}[compare with \cite{Rat}]\label{T-4}
There is a uniquely defined change of variables $\widehat p \to \widehat z\,:\, p_{2k-1} = \beta_k z_{2k-1} + q_{2k-1}(\widehat z)$,
where $\beta_k \in \mathbb{Q}$ and $q_{2k-1}(\widehat z)$ is a homogeneous polynomial, such that
\[
\widehat{\sigma}(z_1,\widehat z) = {\ShW}_g(p_1,\widehat p),\; z_1 = p_1.
\]
\end{thm}

\section*{Thanks}

The authors are grateful to A.\,P.\,Veselov and V.\,N.\,Rubtsov for useful references to works on Adler--Moser polynomials and their generalizations.

\end{document}